\documentclass{article}
\usepackage{url}  %Required
\usepackage{graphicx}  %Required
\usepackage{soul}
\usepackage[utf8]{inputenc}
\usepackage[usenames, dvipsnames]{xcolor}
\usepackage{amsmath}
\usepackage{tikz}
\usepackage{textcomp}
\usepackage{comment} 
\usepackage{amssymb}
\usepackage{xspace}
\usepackage{todonotes}
\usepackage{url}
\usepackage{subfig}

\newcommand{\tilesub}[6]{
	\draw[fill=#3] (#1,#2) -- (#1,#2+1) -- (#1+0.5, #2+0.5) -- (#1, #2);
	\draw[fill=#4] (#1,#2+1)-- (#1+1,#2+1) -- (#1 +0.5, #2+0.5) -- (#1, #2+1);
	\draw[fill=#5] (#1+1,#2)-- (#1+1,#2+1) -- (#1 +0.5, #2+0.5) -- (#1 +1, #2);
	\draw[fill=#6] (#1,#2)-- (#1 +1,#2) -- (#1 +0.5, #2+0.5) -- (#1, #2);
}

\newcommand{\tile}[6]{
	\tilesub{#1}{#2}{#3}{#4}{#5}{#6}

}

\newcommand{\tikztile}[4]{{\tikz[scale=0.5]{\tilesub{0}{0}{#1}{#2}{#3}{#4}}}}
\newcommand{\minitile}[4]{\tikz[scale=0.3]{\tilesub{0}{0}{#1}{#2}{#3}{#4}}}
\newcommand{\tilered}{red!80}
\newcommand{\tileyellow}{yellow!50}
\newcommand{\tilegreen}{green!70}

\newcommand{\tilewhite}{white}

\tikzstyle{tilingsep} = [lightgray]

\newcommand{\squaretilingtop}[2]{
	\draw[fill=#2, draw=none] (#1, 4.5) rectangle (#1+1, 4);}

\newcommand{\squaretilingtopbis}[2]{
	\draw[fill=#2, draw=none] (#1, 5.5) rectangle (#1+1, 5);}

\newcommand{\squaretilingbottom}[2]{
	\draw[fill=#2, draw=none] (#1, -0.5) rectangle (#1+1, 0);}

\newcommand{\squaretilingleft}[2]{
	\draw[fill=#2, draw=none] (-0.5, #1) rectangle (0, #1+1);}

\newcommand{\squaretilingright}[2]{
	\draw[fill=#2, draw=none] (4.5, #1) rectangle (4, #1+1);}

\usetikzlibrary{decorations.pathmorphing}
\usetikzlibrary{decorations.pathreplacing,angles,quotes}
\usetikzlibrary{arrows.meta}

\newcommand\toptile[1]{top_{#1}}
\newcommand\bottomtile[1]{bot_{#1}}
\newcommand\lefttile[1]{lef_{#1}}

\newcommand\righttile[1]{rig_{#1}}
\newcommand{\communication}{\xspace\scalebox{0.5}{\begin{tikzpicture}\draw[ dashed, line width = 1mm ] (0,0) -- (0.8,0);
		\node at (0,-0.08) {}; \end{tikzpicture}}\xspace}
\makeatletter
\tikzset{
	dot diameter/.store in=\dot@diameter,
	dot diameter=3pt,
	dot spacing/.store in=\dot@spacing,
	dot spacing=10pt,
	dots/.style={
		line width=\dot@diameter,
		line cap=round,
		dash pattern=on 0pt off \dot@spacing
	}
}
\makeatother

\tikzstyle{communication} = [color=blue!50!white,dot diameter=1.5pt, dot spacing=1.5pt, dots]

\newtheorem{example}{Example}
\newtheorem{definition}{Definition}
\newtheorem{proposition}{Proposition}
\newtheorem{theorem}{Theorem}
\newtheorem{fact}{Fact}
\newenvironment{proof}{Proof.}{$\square$}

\newcommand\set[1]{\{#1\}}
\newcommand\suchthat{\mid}

\title{Dynamic Connected Cooperative Coverage Problem}
\author{Tristan Charrier \and François Schwarzentruber \and Eva Soulier \\ {Univ Rennes, CNRS, IRISA}}
\begin{document}

\newcommand{\base}{B}

\maketitle
\begin{abstract}
We study the so-called dynamic coverage problem by agents located in some topological graph. The agents must visit all regions of interest but they also should stay connected to the base via multi-hop. We prove that the algorithmic complexity of this planning problem is PSPACE-complete. Furthermore we prove that the problem becomes NP-complete for bounded plans. We also prove the same complexities for the reachability problem of some positions. We also prove that complexities are maintained for a subclass of topological graphs.
\end{abstract}

\section{Introduction}
Unmanned autonomous vehicles (UAVs) are nowadays used in many applications (controlling wildlife, surveying dangerous areas, measuring pollution, etc.).
For example, if a fire occurs, firefighters would send a fleet of UAVs from a base to measure pollution. The UAVs would have then to collaborate so they could map the entire area and always keep communication with the base.

 %unmanned autonomous vehicles (UAVs) have to map or to patrol around a highly constrained geographical area.
As in \cite{DBLP:conf/icc/Yanmaz12}, \cite{teacy2010maintaining} and \cite{IJCAI2018demodrones}, we consider a geographical area, with a launch base  and regions of interest to visit, and topological communication constraints. The big challenge is to synthesize a cooperative plan for the fleet of UAVs for visiting all the regions of interest at least once, always keeping communication with the base and coming back to the base at the end. Communication may be multi-hop (a UAV may communicate to the base via intermediate UAVs). The communication aspect is important in applications such as search-and-rescue.

In this paper, we formally define and study an abstract version of that problem we call it the \emph{dynamic connected cooperative coverage problem}. A geographical area is modeled by a finite graph. The finite graph could be generated from triangulation of the continuous environment (see in Fig. 2, p. 2022  \cite{DBLP:conf/icra/FainekosKP05}, and Fig. 3 of \cite{doi:10.1146/annurev-control-060117-104838}).

 Figure \ref{figure:exampleexecution} shows an execution of an 11-step plan in such a graph: elementary possible moves are represented by solid lines; a dashed line between two regions means that  communication is possible between them.  At the first and final step, UAVs are at the base. At the end, all regions must have been visited. During the execution, UAVs cooperate to stay connected to the base (in other words, dashed lines forming a connected subgraph).

%Even if our problem is idealized, it already contains many aspects of real-life potential scenarios. 
UAVs alternate between moving and performing tasks at all the nodes of the graph that require the UAVs to be stationary %\footnote{For simplicity, we suppose that there is a complex task to perform at all nodes; note that if this is not the case, the problem is even more general and all complexity results in this paper remain true.} 
 (taking high-quality photos, manipulating some objects, etc.). Humans (firemen, engineers, etc.) supervise the UAV mission at the base. Thus, it is required the UAVs to communicate huge amount of data to the base while they are performing tasks. That is why high-speed broadband communications is needed. Such technologies (e.g. laser) typically do not pass through buildings and therefore communication constraints are not trivial to handle. Note that UAVs do not need to communicate huge data to the base while they are moving.

The difficulty resides in the combinatorial when UAVs cooperate to keep communication with the base all along the plan. The plans of the UAVs are inter-dependent.  Even if many UAVs have an ``automated back to launch location'' option, the planning must include the path to come back to launch location 
as we consider using the system in urban areas, where the UAVs would have to avoid any obstacle on their way.

Our problem can be seen as a variation of Multi-Agent Path Finding (MAPF). MAPF consists in finding plans of elementary moves of robots in a grid, starting from an initial situation where each robot has a designated initial cell to a final situation in which each robot has a designated goal cell. The robots should not collide. Finding an optimal plan in the context of MAPF has been proven to be NP-hard (\cite{DBLP:conf/aaai/YuL13}, \cite{DBLP:conf/aaai/MaTSKK16}). MAPF and the dynamic connected cooperative coverage problem differ mainly by their target applications, mainly warehouse or storage robots for the former \cite{DBLP:conf/aaai/WurmanDM07}, search-and-rescue for the latter. Communication and connectivity of UAVs is a main ingredient to our problem compared to MAPF, and, as we will show, it makes our problem computationally more difficult, especially since we do not focus on finding optimal plans - not an optimisation problem - but just finding plans - the existence of a plan.

\begin{figure}[t]
	\begin{center}
		\newcommand{\traitlegendelongueur}{0.3}
		\input{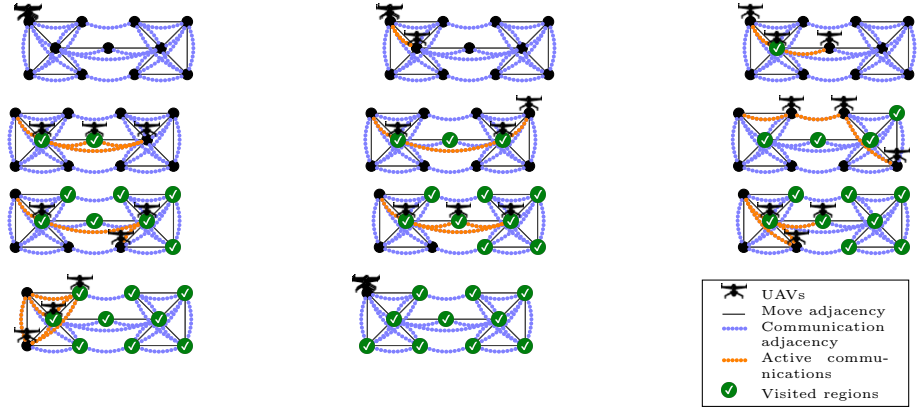}
		\raisebox{-6mm}{
				\tikz{
					\node[inner sep=0.5mm, draw] {\tiny
			\begin{tabular}{p{1mm}p{17mm}}
			\imgDrone & UAVs \\
			\tikz{\draw (0, 0) edge (\traitlegendelongueur, 0);} & Move adjacency \\
			\tikz{\draw[communication] (0, 0) edge (\traitlegendelongueur, 0);} & Communication adjacency \\
			\tikz{\draw[communication, orange] (0, 0) edge (\traitlegendelongueur, 0);} & Active communications \\
			\ok & Visited regions
		\end{tabular}};}
	}
	\end{center}
	\caption{Example of a mission execution.\label{figure:exampleexecution}}
\end{figure}%\todo{en fait, ça fait bête de ne pas avoir de mouvement orienté si dans la suite on dit que c'est orienté. Un exemple où c'est orienté ?}

%For practically synthesizing plans, we advocate two techniques that are in principle from formal methods, namely linear temporal logic (LTL) model checking and classical planning. These two techniques are generic enough to solve any PSPACE problem (\cite{DBLP:journals/jacm/SistlaC85}, \cite{DBLP:journals/ai/Bylander94}). %Both techniques could be used to solve PSPACE problems (\todo{citation LTL PSPACE, planning PSPACE}). , and model checking has already been used to other UAVs scenarios (see \cite{Humphrey2013},\cite{webster2011formal},\cite{DBLP:journals/corr/abs-1003-0381}).
%Therefore, such techniques for easy (typically in P) problems are sledgehammers while other algorithmic techniques may be more efficient.
%So adequate tight lower bounds for the connected cooperative coverage problem, meaning that it is highly combinatorial, justify the need for generic methods.
%

%Our contribution is twofold.
%
%
%\begin{itemize}
%	\item
%
In this paper, we provide theoretical complexity results: we prove that the dynamic connected cooperative coverage problem is PSPACE-complete and that its bounded version is NP-complete. Upper bounds are trivial but lower bounds are proven by delightful reductions from tiling problems \cite{Boas97theconvenience}. It means that synthesizing plans for the dynamic connected cooperative coverage problem is as difficult as classical planning \cite{DBLP:journals/ai/Bylander94}. We also prove that the reachability problem (reaching specific nodes in the graph) has the same complexities. We also prove that the lower bounds are the same even when restricting to a subclass of topological graph for which it is always possible to communicate between two nodes $v$ and $v'$ for which it is possible to move in one step from $v$ to $v'$ (the class of neighbor-communicable topologic graphs).
%	\item

 %Second, we compare the efficiency of LTL model checking and classical planning. We use the model checker NuSMV \cite{cimatti2000nusmv}, an established tool that comes from the verification community and a PDDL planner, FS (for Functional STRIPS) - a recent tool that comes from the formal AI community and is proven to be as efficient as current best planners despite its rich language \cite{DBLP:conf/ijcai/FrancesRLG17}. To sum up, we show that planning is more efficient than model checking for bigger instances, but does not have the same potential for extensions.%, yet it has interest to keep model checking because of its modulability. \todo{Relire ce paragraphe, franchement je vois pas comment faire mieux. }
%\end{itemize}

 %Section~\ref{section:definition} settles the definitions. Section \ref{section:complexityupperbounds} provides the theoretical complexity upper bounds. Section~\ref{section:tilings} recalls basics about tiling problems. Section~\ref{section:complexitylowerbounds} provides the theoretical complexity lower bounds, we obtained via reduction from tiling problems. Finally Section~\ref{section:relatedwork} details related work.
 The rest of the paper is organized as follows. First we settle the definition. Second, we  provide the theoretical complexity upper bounds. Third, we recall basics about tiling problems and then provide the theoretical complexity lower bounds, we obtained via reduction from tiling problems. Finally we detail related work.

\section{Definitions}
\label{section:definition}
\subsection{Topologic graph}
\label{subsection:topologicalgraph}
A geographical area is modeled by a \emph{topologic graph}. Nodes are regions of interest where the launch base is  a special region noted $\base$. Relation $\rightarrow$ represents possible moves of UAVs: $v \rightarrow v'$ if a UAV can reach $v'$ from $v$ in one step. Relation $\communication$ represents possible communications: $v \communication v'$ if any UAV at $v$ can communicate with any UAV at $v'$. We say that $v$ \emph{communicates} with $v'$. Formally:

\begin{definition}[Topological graph]
\label{definition:topologicgraph}
A \emph{topologic graph} is a tuple $G$=($V$,$\rightarrow$,$\communication$) where $V$ is a non-empty finite set of regions containing a specific element $\base$ and $\rightarrow,\communication \, \subseteq V \times V$ is such that $(V, \communication)$ is a non-oriented graph, $\base \rightarrow \base$.
\end{definition}

%\begin{example}
 Definition \ref{definition:topologicgraph} imposes $\base$ to be $\rightarrow$-reflexive since a UAV can stay at the base. Notice that it does not impose the relation $\rightarrow$ to be symmetric or other nodes to be reflexive so that we can capture windy environments, one-sided roads etc.

For sufficiently fine-grained topologic graphs, there are no obstacles between regions $v$ and $v'$ when $v \rightarrow v'$. Thus, communications between $v$ and $v'$ are not perturbed and $v \communication v'$. In words, if one UAV can reach $v'$ from $v$ in one step then a communication between two UAVs in $v$ and $v'$ is possible. This hypothesis seems reasonable for many means of communication (lasers, etc.). That is why, we define the subclass of \emph{neighbor-communicable} topologic graphs in which % $\rightarrow \subseteq \communication$, that is, when
 $v \rightarrow v'$ implies $v \communication v'$. 
%In practice, if the situation is represented by a topologic graph with enough nodes,  neighbor nodes do not have obstacles between them and are close enough, which ensures the communication. Therefore it is usually the case that graphs are neighbor-communicable \todo{relire} Formally a topologic graph is \emph{neighbor-communicable} if $\rightarrow \subseteq \communication$, that is, when $v \rightarrow v'$ implies $v \communication v'$. 

\subsection{Executions}

%Indeed, one could imagine that it is possible to move from $v$ to $v'$ but not to come back to $v$ from $v'$ because there is too much wind and it would cost too much energy. We could easily label $\rightarrow$-transitions with energy costs but we are not going to do so to keep our topologic graph definition simple.

Given a topologic graph and given $n$ UAVs, a configuration $c$ gives positions to each UAV such that they form a multi-hop system: they are all connected to the base. Furthermore, we suppose that at most one UAV is at a given region, except at~$B$.
\newcommand{\setofregionsUAVs}{V_{\text{\tiny UAVs}}}
\begin{definition}[Configuration]
	\label{definition:configuration}
A \emph{configuration} $c$ is an element of $V^n$ such that the graph $(\setofregionsUAVs, \communication \cap \setofregionsUAVs \times \setofregionsUAVs)$ is connected with $\setofregionsUAVs =\{c_i ~/~ i \in \{1,..n\}\} \cup \{B\}$ and for all $i \neq j$, if $c_i \neq B$, then  $c_i \neq c_j$. We note $c \rightarrow c'$ when $c_i \rightarrow c_i' \in G$ for all $i \in \{1,..n\}$.
\end{definition}

Without loss of generality, as the UAVs are interchangeable, configurations are equivalent up to a permutation of UAVs. To avoid cumbersome notations in proofs, we consider equivalent configurations as equal. For instance, for $n = 5$, configurations $(B, B, v, v', v'')$ and $(B, v, v'', B, v')$ are equivalent.

\begin{definition}[Execution] An \emph{execution} in $G$ with $n$ UAVs of length $\ell$ is a sequence of configurations $(c^0,...,c^\ell)$ such that $c^0 \rightarrow c^1\rightarrow ... \rightarrow c^\ell$. A \emph{covering execution} in $G$ with $n$ UAVs of length $\ell$ is an execution $(c^0,...,c^\ell)$ such that $c^0 = c^\ell = (B,B,...,B)$ (all UAVs are at B at the start and at the end) and $\{c_i^t ~/~ t \in \{0,..\ell\}, ~ i \in \{1,..n\}\}=V$ (all regions are visited at some point).
\end{definition}

Notice that if we have $c_1 \rightarrow c_1'$ and $c_1$ and $c_2$ are equivalent then by taking $c_2'$ the same permutation of $c_1'$ compared to $c_1$ and $c_2$, then $c_2 \rightarrow c_2'$. Thus, it is always possible to transform an execution into an equivalent one by a permutation. % \todo{Relire} 
 \begin{example}
 	Figure \ref{figure:exampleexecution} shows a topologic graph with 11 regions. Here, the $\rightarrow$-relation is symmetric and is represented by solid black lines. The $\communication$-relation is represented by dotted lines and is blue when not taken, orange when taken. Visited regions are represented by checked marks. The execution is read line by line. Notice that at each step, the subgraph $(\setofregionsUAVs, \communication \cap \setofregionsUAVs \times \setofregionsUAVs)$ is connected; it is the subgraph obtained by taking only the active communication lines in Figure \ref{figure:exampleexecution}.  Although the topological graph has 11 nodes, it is sufficient to have 3 UAVs to map the topological graph.
 \end{example}%, i.e. $\forall v \in V, \exists$ an UAV $i \in \{1,..n\} ~/~$ at some time $t,~c_i^t =v.$

% Item 1 means that UAVs are all at $B$ at the start and the end of the execution. Item 2 means that all regions in $V$ are visited by some UAV $i$ at some time $t$.

\newcommand{\problemcoverage}{\textbf{Coverage}\xspace}
\newcommand{\problemreachability}{\textbf{Reachability}\xspace}
\newcommand{\problemboundedcoverage}{\textbf{bCoverage}\xspace}
\newcommand{\problemboundedreachability}{\textbf{bReachability}\xspace}

\subsection{Decision problems}

We define the connected cooperative coverage problem and the connected cooperative reachability problem shortly denoted by \problemcoverage and \problemreachability. The reachability problem essentially is introduced for pedagogical reasons, especially for making the lower bound results more diligent. We also define bounded versions of the two decision problems, namely \problemboundedcoverage and \problemboundedreachability. The bounded versions are inspired from the so-called polynomial-length planning problem \cite{DBLP:conf/jelia/Turner02} in which we ask for the existence of a plan of length bounded by a polynomial in the size of the planning task. It is equivalent to add the length bound as an input to the decision problems in unary. Formally:

\newcommand{\definitiondecisionproblem}[2]{\vspace{-2mm}\begin{itemize}
			\item Input: #1;
			\item Output: yes if #2; no otherwise.
	\end{itemize}}

\begin{definition}[Coverage problems] ~

	\problemcoverage: \definitiondecisionproblem{a topologic graph $G$ and $n \in \mathbb{N}$}{there is a covering execution in $G$ with $n$ UAVs}
	
	\problemboundedcoverage: \definitiondecisionproblem{a topologic graph $G$, $n \in \mathbb{N}$ and $\ell \in \mathbb{N}$ in unary}{there is a covering execution in $G$ with $n$ UAVs of length at most $\ell$}
\end{definition}

%	\problemcoverage is defined by:
%	\begin{itemize}
%		\item Input: a topologic graph $G$, $n \in \mathbb{N}$;
%		\item Output: yes if  there exists a covering execution in $G$ with $n$ UAVs; no otherwise.
%	\end{itemize}

\begin{definition}[Reachability problems] ~

\problemreachability:  \definitiondecisionproblem{a topologic graph $G$ and a configuration $c$}{there is an execution $(c^0, \dots, c^\ell)$ in $G$ such that $c^0 = (B, \dots, B)$ and $c^\ell = c$}
%	\problemreachability is defined by:
%	\begin{itemize}
%		\item Input: a topologic graph $G$, a configuration $c$;
%		\item Output: yes if  there exists an execution $(c^0,...,c^\ell)$ in $G$ such that $c^0 = (B, \dots, B)$ and $c^\ell = c$; no otherwise.
%	\end{itemize}

%	\problemboundedcoverage is defined by:
%	\begin{itemize}
%		\item Input: a topologic graph $G$, $n \in \mathbb{N}$, $\ell \in \mathbb{N}$ where $\ell$ is given in unary;
%		\item Output: yes if  there exists a covering execution in $G$ with $n$ UAVs of length at most $\ell $; no otherwise.
%	\end{itemize}

\problemboundedreachability: \definitiondecisionproblem{a topologic graph $G$, a configuration $c$ and $\ell \in \mathbb{N}$ in unary}{there is an execution $(c^0, \dots, c^{\ell'})$ in $G$ such that $c^0 = (B, \dots, B)$, $c^{\ell'} = c$ and $\ell' \leq \ell$}

\end{definition}

We now establish upper bound complexities of \problemcoverage, \problemboundedcoverage, \problemreachability and \problemboundedreachability.
%	\problemboundedreachability is defined by:
%	\begin{itemize}
%		\item Input: a topologic graph $G$, a configuration $c$, $\ell \in \mathbb{N}$ where $\ell$ is given in unary;
%		\item Output: yes if  there exists an execution $(c^0,...,c^{\ell'})$ in $G$ such that $c^0 = (B, \dots, B)$ and $c^{\ell'} = c$ and $\ell' \leq \ell$; no otherwise.
%	\end{itemize}

\section{Complexity: upper bounds}
\label{section:complexityupperbounds}
\begin{proposition}
	\problemcoverage and \problemreachability are in PSPACE.\label{proposition:problemcoveragePSPACE}
	\end{proposition}
	
	\begin{proof}
		In both cases, the straightforward non-deterministic guessing an execution runs in polynomial space: for \problemcoverage,  we only keep in memory the last configuration and the set of already visited regions. For \problemreachability, we only keep in memory the last configuration. By Savitch's theorem (NPSPACE = PSPACE) \cite{DBLP:journals/jcss/Savitch70}, the proposition is proven. 
	\end{proof}
	
\begin{proposition}
	\problemboundedcoverage and \problemboundedreachability are in NP.
\end{proposition}

	\begin{proof}
		We define the same algorithms given in the Proof of Proposition~\ref{proposition:problemcoveragePSPACE} except that we stop the execution when the length is exceeded. Thus, the algorithms are non-deterministic and run in polynomial time.
	\end{proof}

\section{Tiling problems}
\label{section:tilings}
\newcommand{\constrainthorizontally}{(h)\xspace}
\newcommand{\constraintvertically}{(v)\xspace}

Tilings were introduced by Wang (\cite{Wang1961,Wang1990}).
As pointed out by van der Boas (\cite{van1984boundedtiling}, \cite{Boas97theconvenience}), tilings offer convenient decision problems for proving lower bound complexity. We also cite Levin's work who invented the notion of NP-completeness independently from Cook and who introduced a bounded tiling problem \cite{levin1973}. Some tiling problems are also addressed in some textbooks to characterize some complexity classes (\cite{DBLP:books/daglib/0096996}, p. 262, 310; \cite{Harel:2000:DL:557365}, p. 58-63).
 We use \emph{tile types} $t$ that are tuples $\langle left(t), up(t), right(t), down(t)\rangle \in \mathbb{N}^4$ giving colors (represented by integers) to the four sides of a tile \minitile{\tilewhite}{\tilewhite}{\tilewhite}{\tilewhite}. A tiling is represented by a function $\lambda$ that maps a tile type to each position $(i, j)$. Two horizontally or vertically adjacent tiles should match horizontally (constraints \constrainthorizontally and \constraintvertically in the following Definitions).  The two decision problems introduced in this section are taken from \cite{Boas97theconvenience}.

 \subsection{Square tiling problem}
 
 \newcommand{\drawsquaretilingproblemexamplecadre}{
 	\draw (0, 0) rectangle (4, 4);
 	
 	\squaretilingtop 0 {\tilered}
 	\squaretilingtop 1 {\tileyellow}
 	\squaretilingtop 2 {\tilered}
 	\squaretilingtop 3 {\tilegreen}
 	
 	\squaretilingbottom 0 {\tilered}
 	\squaretilingbottom 1 {\tileyellow}
 	\squaretilingbottom 2 {\tilered}
 	\squaretilingbottom 3 {\tilegreen}
 	
 	\squaretilingleft 0 {\tilered}
 	\squaretilingleft 1 {\tileyellow}
 	\squaretilingleft 2 {\tilered}
 	\squaretilingleft 3 {\tilegreen}
 	
 	\squaretilingright 0 {\tilered}
 	\squaretilingright 1 {\tileyellow}
 	\squaretilingright 2 {\tilered}
 	\squaretilingright 3 {\tilegreen}

 	\foreach \x in {1, 2, 3} {
 		\draw[tilingsep] (\x, 0) -- (\x, 4);
 		\draw[tilingsep] (0, \x) -- (4, \x);
 	}
 	}

 \newcommand{\drawcorridortilingproblemexamplecadre}{
 	\draw (0, 0) rectangle (3, 4);
 	
 	\squaretilingtop 0 {\tileyellow}
 	\squaretilingtop 1 {\tilegreen}
 	\squaretilingtop 2 {\tilered}

 	\squaretilingbottom 0 {\tilered}
 	\squaretilingbottom 1 {\tileyellow}
 	\squaretilingbottom 2 {\tilegreen}

 }
   \newcommand{\drawcorridortilingproblemexamplecadrebis}{
   	\draw (0, 0) rectangle (3, 5);
   	
 	\squaretilingtopbis 0 {\tileyellow}
 	\squaretilingtopbis 1 {\tilegreen}
 	\squaretilingtopbis 2 {\tilered}

 	\squaretilingbottom 0 {\tilered}
 	\squaretilingbottom 1 {\tileyellow}
 	\squaretilingbottom 2 {\tilegreen}

   	\foreach \x in {1, 2, 3, 4} {
   		\draw[tilingsep] (0, \x) -- (3, \x);
   	}
   	\foreach \x in {1, 2} {   	
   	   		\draw[tilingsep] (\x, 0) -- (\x, 4);
   	   	}
   }
 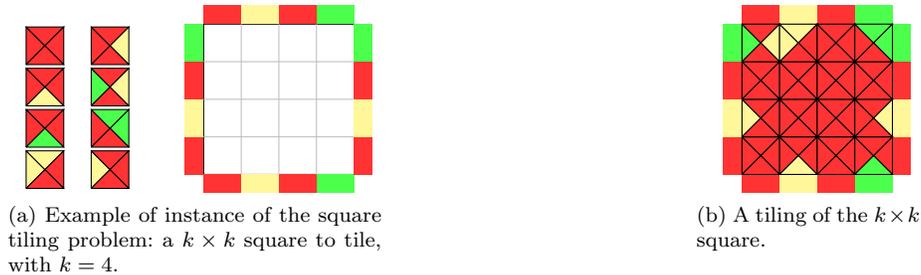
\begin{figure}[t]
 	\centering
 	
 	\subfloat[Example of instance of the square tiling problem: a $k \times k$ square to tile, with $k = 4$.]
 	{\label{figure:squaretilingproblem}
 		\raisebox{8mm}{
 		\begin{minipage}{20mm}
 			\tikztile\tilered\tilered\tilered\tilered
 ~~
 			\tikztile\tilered\tilered\tileyellow\tilered
 			 			
 			\tikztile\tilered\tilered\tilered\tileyellow
 	~~		
 			\tikztile\tilegreen\tilered\tileyellow\tilered
 			
 			\tikztile\tilered\tilered\tilered\tilegreen
 	~~		
 			\tikztile\tilered\tilegreen\tilegreen\tilered
 			
 			\tikztile\tileyellow\tileyellow\tilered\tilered
 	~~		
 			\tikztile\tileyellow\tilered\tilered\tilered
 		\end{minipage}}
 		\begin{tikzpicture}[scale=0.5, baseline=0mm]
 		\drawsquaretilingproblemexamplecadre
  		\end{tikzpicture}
 	}
 	\hfill
 	\subfloat[A tiling of the $k \times k$ square.]
 	{\label{figure:squaretilingproblemsolution}
 	%	\scalebox{1}
 		{
 		\raisebox{-2.5mm}{
 		\begin{tikzpicture}[scale=0.5]
 		\drawsquaretilingproblemexamplecadre
 		
 		%x y left up right down
 		\tile00 \tilered\tilered\tilered\tilered
 		\tile10 \tilered\tilered\tilered\tileyellow
 		\tile20 \tilered\tilered\tilered\tilered
 		\tile30 \tilered\tilered\tilered\tilegreen
 		
 		\tile01 \tileyellow\tilered\tilered\tilered
 		\tile11 \tilered\tilered\tilered\tilered
 		\tile21 \tilered\tilered\tilered\tilered
 		\tile31 \tilered\tilered\tileyellow\tilered
 		
 		\tile02 \tilered\tilered\tilered\tilered
 		\tile12 \tilered\tilered\tilered\tilered
 		\tile22 \tilered\tilered\tilered\tilered
 		\tile32 \tilered\tilered\tilered\tilered
 		
 		\tile03 \tilegreen\tilered\tileyellow\tilered
 		\tile13 \tileyellow\tileyellow\tilered\tilered
 		\tile23 \tilered\tilered\tilered\tilered
 		\tile33 \tilered\tilegreen\tilegreen\tilered
 		\end{tikzpicture}}}
 	}
 	\caption{The square tiling problem.}
 \end{figure}
 
 The square tiling problem consists in tiling a $k \times k$ square as depicted Figure~\ref{figure:squaretilingproblem}, by using finite set of tile types and by respecting boundary color constraints along the edges. Figure~\ref{figure:squaretilingproblemsolution} shows such a tiling. Note that tiles cannot be turned. Formally:
 
 \newcommand{\colorsequencetop}[1]{top_{#1}}
 \newcommand{\colorsequencebottom}[1]{bot_{#1}} 
 \newcommand{\colorsequenceleft}[1]{left_{#1}}
\newcommand{\colorsequenceright}[1]{right_{#1}}

 \begin{definition}[Square tiling problem]
 	 The square tiling problem is the following decision problem:
 	 \begin{itemize}
 	 	\item Input: a set $T \subseteq \mathbb{N}^4$  of tiles types and four sequences $\colorsequencetop1, \dots, \colorsequencetop k \in T$, $\colorsequencebottom1, \dots, \colorsequencebottom k \in T$, $\colorsequenceleft 1, \dots, \colorsequenceleft k \in T$, $\colorsequenceright 1, \dots, \colorsequenceright k \in T$ of length $k$
 	 \item Output :yes if there is a function $\lambda: \{1,..,k\}\times\{1,..k\}\rightarrow T$ such that:
 	 	\begin{enumerate}
 	 		\item[\constrainthorizontally] $right(\lambda(i,j)) = left(\lambda(i+1,j))$ for all $i \in \{1,..k-1\}$, for all $j \in \{1,..m\}$;
 	 		\item[\constraintvertically] $up(\lambda(i,j)) = down(\lambda(i,j+1))$ for all $i \in \{1,..k\}$, for all $j \in \{1,..m-1\}$;
 	 		\item $up(\lambda(1,j)) = \colorsequencetop j$ for all $j \in \{1,..k\}$;
 	 		\item $down(\lambda(k, j))  = \colorsequencebottom j$ for all $j \in \{1,..k\}$;
 	 		\item $left(\lambda(i,1)) = \colorsequenceleft i$ for all $i \in \{1,..k\}$;
 	 		\item $right(\lambda(i,k)) = \colorsequenceright i$  for all $i \in \{1,..k\}$.
 	 	\end{enumerate}
 	 	no otherwise
 	 	\end{itemize}
 \end{definition}

  \begin{theorem}\label{theorem:squaretilingproblemPSPACEcomplete}
 	The square tiling problem is NP-complete \cite{Boas97theconvenience,van1984boundedtiling}. 
 \end{theorem}

\subsection{Corridor tiling problem}
 
Contrary to the square tiling problem, the corridor tiling problem consists in tiling a $k \times m$-rectangle, where $m$ is arbitrary, by respecting the top and bottom edge constraints, left and right edges being all white. Formally:
 
% \begin{definition}[Corridor tiling problem] Given a finite set $T \subseteq \mathbb{N}^4$  of tiles types and two sequences $\toptile{1},...,\toptile{k}$ and $\bottomtile{1},...,\bottomtile{k}  \in T$ of length $k$ representing the top line and bottom line of tiles (e.g. Figure~\ref{fig:extiles}), is there a height $m$ and a tiling of the $k \times m$ rectangle whose top and bottom edges are $\toptile{1},...,\toptile{k}$ and $\bottomtile{1},...,\bottomtile{k}$ respectively, and the left and right edges are all white.	
% \end{definition}
% 
\begin{definition}[Corridor tiling problem] The corridor tiling problem is the following decision problem:
\begin{itemize}
 \item input: A set $T \subseteq \mathbb{N}^4$  of tiles types and two sequences $\toptile{1},...,\toptile{k}$ and $\bottomtile{1},...,\bottomtile{k}  \in T$ of length $k$;% which respectively represents the top line and bottom line of tiles.
 \item output: yes if there exists  an integer $m$ and a function $\lambda: \{1,..,m\}\times\{1,..k\}\rightarrow T$ such that:
 \begin{enumerate}
\item[\constrainthorizontally] $right(\lambda(i,j)) = left(\lambda(i+1,j))$ for all $i \in \{1,..k-1\}$, for all $j \in \{1,..m\}$;
\item[\constraintvertically] $up(\lambda(i,j)) = down(\lambda(i,j+1))$ for all $i \in \{1,..k\}$, for all $j \in \{1,..m-1\}$;
\item $\lambda(1,j) =\bottomtile j$ for all $j \in \{1,..k\}$;
\item $\lambda(m,j) = \toptile j$ for all $j \in \{1,..k\}$;
\item $left(\lambda(m,1)) = right(\lambda(m, k)) = white$.
 \end{enumerate}
\end{itemize}
\end{definition}

 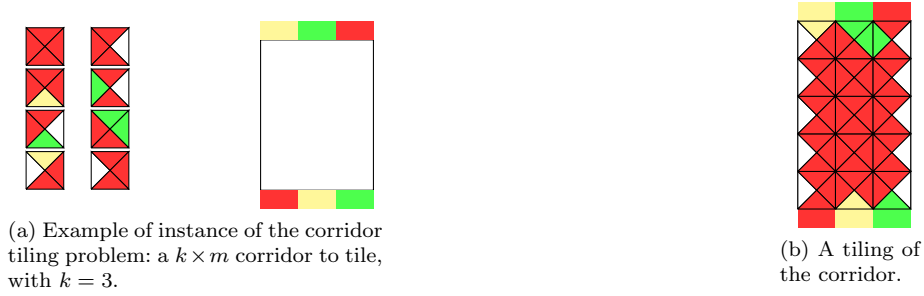
\begin{figure}[t]
 	
 	 	\centering
 	 	
 	 	\subfloat[Example of instance of the corridor tiling problem: a $k \times m$ corridor to tile, with $k = 3$.]
 	 	{\label{figure:squaretilingproblem}
 	 		\raisebox{1cm}{
 	 			\begin{minipage}{3cm}
 	 				\tikztile\tilered\tilered\tilered\tilered
 	 				~~
 	 				\tikztile\tilered\tilered\tilewhite\tilered
 	 				
 	 				\tikztile\tilered\tilered\tilered\tileyellow
 	 				~~		
 	 				\tikztile\tilegreen\tilered\tilewhite\tilered
 	 				
 	 				\tikztile\tilered\tilered\tilewhite\tilegreen
 	 				~~		
 	 				\tikztile\tilered\tilegreen\tilegreen\tilered
 	 				
 	 				\tikztile\tilewhite\tileyellow\tilered\tilered
 	 				~~		
 	 				\tikztile\tilewhite\tilered\tilered\tilered
 	 			\end{minipage}}
 	 			\begin{tikzpicture}[scale=0.5, baseline=0mm]
 	 			\drawcorridortilingproblemexamplecadre
 	 			\end{tikzpicture}
 	 		}
 	 		\hfill
 	 		\subfloat[A tiling of the corridor.]
 	 		{\label{figure:squaretilingproblemsolution}
 	 			%	\scalebox{1}
 	 			{
 	 				\raisebox{-0.5cm}{\begin{tikzpicture}[scale=0.5]
 	 				\drawcorridortilingproblemexamplecadrebis
 	 				
 	 				%x y left up right down
 	 				\tile00 \tilewhite\tilered\tilered\tilered
 	 				\tile10 \tilered\tilered\tilered\tileyellow
 	 				\tile20 \tilered\tilered\tilewhite\tilegreen
 	 				
 	 				\tile01 \tilewhite\tilered\tilered\tilered
 	 				\tile11 \tilered\tilered\tilered\tilered
 	 				\tile21  \tilered\tilered\tilewhite\tilered
 	 				
 	 				\tile02 \tilewhite\tilered\tilered\tilered
 	 				\tile12 \tilered\tilered\tilered\tilered
 	 				\tile22 \tilered\tilered\tilewhite\tilered
				
 	 				\tile03 \tilewhite\tilered\tilered\tilered
 	 				\tile13 \tilered\tilered\tilered\tilered
 	 				\tile23 \tilered\tilered\tilewhite\tilered

 	 				\tile04 \tilewhite\tileyellow\tilered\tilered 
 	 				\tile14\tilered\tilegreen\tilegreen\tilered
 	 				\tile24 \tilegreen\tilered\tilewhite\tilered

 	 				\end{tikzpicture}}}
 	 		}
 	\caption{The corridor tiling problem.\label{figure:coridortilingproblem}}
 \end{figure}

 %
 %The corridor tiling problem is the following PSPACE-complete problem:
 %\begin{itemize}
 % \item input: A set $T \subseteq \mathbb{N}^4$  of tiles types and two sequences $\toptile{1},...,\toptile{k}$ and $\bottomtile{1},...,\bottomtile{k}  \in T$ of length $k$;% which respectively represents the top line and bottom line of tiles.
 % \item output: yes if there exists a height $m$ and a tiling of the $k \times m$ rectangle whose top and bottom 
 %edges are $\toptile{1},...,\toptile{k}$ and $\bottomtile{1},...,\bottomtile{k}$ respectively, as depicted in figure \ref{fig:extiles}; no otherwise.
 %%? The answer is yes if $\exists m \in \mathbb{N},~ \lambda: \{1,..n\}\times\{1,..k\}\rightarrow T$
 %%
 %%$\lambda(1,j) = b_j ~~ \forall j \in \{1,..k\}$
 %%
 %%$\lambda(m,j) = u_j ~~ \forall j \in \{1,..k\}$
 %%
 %%$right(\lambda(i,j)) = left(\lambda(i+1,j))~~ \forall i \in \{1,..k-1\},~ \forall j \in \{1,..m\}$
 %%
 %%$up(\lambda(i,j)) = down(\lambda(i,j+1)) ~~ \forall i \in \{1,..k\}, ~ \forall j \in \{1,..m-1\}$
 %%\todo{ajouter cotés blancs}
 %%on verra plus tard si on besoin du truc formel
 %\end{itemize}

 \begin{theorem}\label{theorem:corridortilingproblemPSPACEcomplete}
 	The corridor tiling problem is PSPACE-complete \cite{Boas97theconvenience}.
 \end{theorem}

\section{Complexity: lower bounds}
\label{section:complexitylowerbounds}
\newcommand{\colorellipse}{orange!20!white}

\subsection{PSPACE lower bounds}

In this subsection, we reduce the corridor tiling problem that is PSPACE-complete (Theorem~\ref{theorem:corridortilingproblemPSPACEcomplete}) to $\problemreachability$ and $\problemcoverage$. First we start with \problemreachability. Independently, a similar reachability problem, without a base, was proven PSPACE-hard in \cite{DBLP:conf/aaai/TateoBRAB18}. Their proof relies on Nondeterministic Constraint Logic  \cite{DBLP:conf/coco/DemaineH08}.
%In this section, we show results on the lower bound of our UAVs problems by means of a reduction from tiling problems.

\begin{theorem}
	\problemreachability is PSPACE-hard.
	\label{theorem:reachabilitypspacehard}
\end{theorem}

\begin{proof}
The  proof is by polynomial time reduction from the corridor tiling problem. %This will show that our problem is PSPACE-hard.
To do so, we map a desired corridor tiling instance\linebreak[4] $(T,\toptile{1},...,\toptile{k},\bottomtile{1},...,\bottomtile{k})$ to the \problemreachability instance $(G, k, c)$ described below.

\paragraph{Description of $G$. }
The topologic graph is shown in Figure~\ref{reduction}. The set of nodes in $G$ contains the base $B$, a copy of $T$ with only tiles with white left-side, $k-2$ copies of $T$ and a copy of $T$ with only tiles with white right-side. Copies of $T$ are represented by ellipses in Figure~\ref{reduction}.
%In the ellipses, we can find copies of the set of tiles $T$.

%Note that the nodes in the $k$ ellipses and the nodes $\toptile{1},...,\toptile{k},\bottomtile{1},...,\bottomtile{k}$ of the topologic graph are tiles. 
Possible moves for the UAVs are represented by the arrows~$\rightarrow$ in Figure~\ref{reduction}. Moreover, a UAV can also move
from tile $t$ to tile $t'$ when $up(t) = down(t')$, i.e.
                 $\minitile{\tilewhite}{\tilered}{\tilewhite}{\tilewhite} \rightarrow \minitile{\tilewhite}{\tilewhite}{\tilewhite}{\tilered}$, and $t$ and $t'$ belong to the same copy of $T$.
                
Communication links between regions are represented by \textcolor{blue!50!white}{$\communication$} in Figure \ref{reduction}. Moreover, a UAV on tile $t$ in the $i^{\text{th}}$ ellipse can communicate with another UAV on tile $t'$ in the $i+1^{\text{th}}$ ellipse if $right(t)=left(t')$, i.e. 
%\begin{tikzpicture}[scale=0.4]
%\draw (0,0) rectangle (1,1);
%\draw (0,0) -- (1,1);
%\draw (0,1) -- (1,0);
%\fill[gray] (1,0) -- (1,1) -- (0.5,0.5);
%\draw (2,0) rectangle (3,1);
%\draw (2,0) -- (3,1);
%\draw (2,1) -- (3,0);
%\fill[gray] (2,0) -- (2,1) -- (2.5,0.5);
%\draw[<->,communication] (1.1,0.5) -- (1.9,0.5);
%\end{tikzpicture}  
%
$\minitile{\tilewhite}{\tilewhite}{\tilered}{\tilewhite}$ \textcolor{blue!50!white}{$\communication$} $\minitile{\tilered}{\tilewhite}{\tilewhite}{\tilewhite}$.

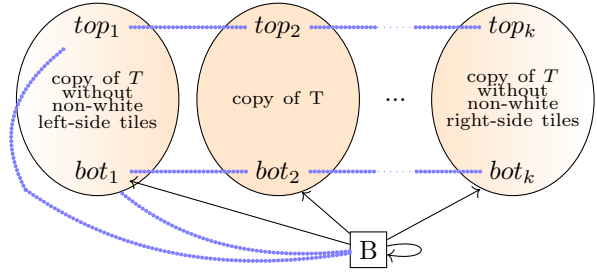
\begin{figure}[t]
	\begin{center}
		\newcommand{\xdots}{4.3}
		\newcommand{\dotswidth}{0.2}
		\newcommand{\xcopyk}{5.6}
		\begin{tikzpicture}[xscale=1.2,yscale=0.8]
		\draw[left color=white, right color=\colorellipse] (1,3) ellipse (0.9 and 1.6);
		\draw[fill=\colorellipse] (3,3) ellipse (0.9 and 1.6);
		%	\draw (5,3) ellipse (0.9 and 2.1);
		\draw[left color=\colorellipse, right color=white] (\xcopyk,3) ellipse (0.9 and 1.6);
		\draw (4,0.5) node[draw] (nb) {B};
		\draw (1,1.8) node (n1) {$\bottomtile{1}$};
		\draw (3,1.8) node (n2) {$\bottomtile{2}$};
		\draw (\xcopyk,1.8) node (nk) {$\bottomtile{k}$};
		\draw (1,4.2) node (t1) {$\toptile{1}$};
		\draw (3,4.2) node (t2) {$\toptile{2}$};
		\draw (\xcopyk,4.2) node (tk) {$\toptile{k}$};
		\draw[->] (nb) -- (n1);
		\draw[->] (nb) -- (n2);
		\draw[->] (nb) -- (nk);
		\draw (1,3) node [text width = 1.8cm, text centered] {\scriptsize copy of $T$ \\[-1mm] without  \\[-1mm] non-white  \\[-2mm] left-side tiles};
		\draw (3,3) node [text width = 1.8cm, text centered] {\scriptsize{copy of T}};
		\draw (\xdots,3) node [text width = 1.8cm, text centered] {...};
		\draw (\xcopyk,3) node [text width = 1.8cm, text centered] {\scriptsize copy of $T$ \\[-1mm] without  \\[-1mm] non-white  \\[-2mm] right-side tiles};
		%\draw[<->,communication] (nb) .. controls (0,-1) and (-0.5,1.5) .. (n1);
		\draw[communication] (nb) to[bend left] (n1);
		\draw[communication] (nb) to[bend left] (0.2,1.5);
		\draw[communication] (0.2,1.5) to[bend left] (t1);
		\draw[communication] (n2) -- (n1);
		\draw[communication] (n2) -- (\xdots-\dotswidth,1.8);
		\draw[dotted,color=blue!50!white] (\xdots-\dotswidth,1.8) -- (\xdots+\dotswidth,1.8);
		\draw[communication] (nk) -- (\xdots+\dotswidth,1.8);
		\draw[communication] (t2) -- (t1);
		\draw[communication] (t2) -- (\xdots-\dotswidth,4.2);
		\draw[dotted,color=blue!50!white] (\xdots-\dotswidth,4.2) -- (\xdots+\dotswidth,4.2);
		\draw[communication] (tk) -- (\xdots+\dotswidth,4.2);
		\draw[->] (nb) to [loop right] (nb);
		\end{tikzpicture}
		\caption{Topologic graph of the corridor tiling problem reduction.}
		\label{reduction}
	\end{center}
\end{figure}
More formally $G$ = ($V$,$\rightarrow$,$\communication$) is defined by:
\begin{itemize}
	\item $V$ is the disjoint union of $\set{B}$, $\set{(t, 1) \suchthat t \in T \text{ and } left(t) = white}$, $\set{(t, j) \suchthat t \in T \text{ and } j \in \set{2, \dots, k-1}}$ and $\set{(t, k) \suchthat t \in T \text{ and } right(t) = white}$;
	\item $\rightarrow$ is the union of $\set{(B, (\bottomtile i, i)) \suchthat i \in \set{1, \dots, k}}$ and $\set{((t, i), (t', i)) \suchthat up(t) = down(t')}$ for all $i \in \set{1, \dots, k}$;
	\item $\communication$ is the union of $\set{(B, (t, 1)) \suchthat (t, 1) \in V}$ and $\set{((t, i), (t', i+1)) \suchthat right(t) = left(t')}$ for all $i \in \set{1, \dots, k-1}$.
		
\end{itemize}

The formal definition $G$ sums up the informal explanation given above. Tile $t$ in the $i^{\text{th}}$ ellipse is denoted by $(t, i)$.

The goal configuration is $c = (\toptile{1},...,\toptile{k})$.

\paragraph{Intuition. } 

%A tiling of the $k \times m$ rectangle whose top and bottom 
%edges are $\toptile{1},...,\toptile{k}$ and $\bottomtile{1},...,\bottomtile{k}$ respectively exists if and only if
%there exists a covering execution in $G$ with $k$ UAVs.
%\begin{itemize}
%\item 

The intuition is that once all the UAVs leaved the base, a configuration corresponds to a row of $k$ tiles in the $k \times m$ rectangle in Figure \ref{figure:coridortilingproblem}. The  position of the UAV in $i^\text{th}$ ellipse corresponds to the tile in column $i$. In such a row, tiles match horizontally by definition of \textcolor{blue!50!white}{$\communication$}. The second configuration $\toptile{1},...,\toptile{k}$ corresponds to the bottom row. Each time that the execution progresses, UAVs synchronously move in new tiles: it mimics 
a new row added to the tiling in construction. Tiles match vertically by definition of $\rightarrow$-transitions.

A tiling of the $k \times m$ rectangle whose top and bottom 
edges are $\toptile{1},...,\toptile{k}$ and $\bottomtile{1},...,\bottomtile{k}$ respectively exists if and only if the
UAVs can reach the configuration $(\toptile{1},...,\toptile{k})$. 
%
%($\Rightarrow$). Suppose there exists a tiling  of the $k \times m$ rectangle whose top and bottom 
%	edges are $\toptile{1},...,\toptile{k}$ and $\bottomtile{1},...,\bottomtile{k}$ respectively. We construct an execution that reaches $c$.
%	
%($\Leftarrow$) If there is an execution that reaches $c$, then we extract from it a tiling of the $k \times m$ rectangle with $m$ the length of the path of a UAV in the ellipses area.
%
%
%
\end{proof}

\tikzstyle{copyofG} = [fill=gray!20!white, draw=none, rounded corners=8pt]

\newcommand{\addCommunicationB}[1]
{
	\node[above right = 0.5mm and 4mm of #1, inner sep=0mm] (nb2) {B};
	\draw[communication] (#1) -- (nb2);
}

\begin{theorem}
	\problemcoverage is PSPACE-hard.
	\label{theorem:coveragepspacehard}
\end{theorem}

\begin{proof}
	The proof is by reduction from \problemreachability. To do so we map an instance  $(G, k, c)$ of \problemreachability to the instance $(G', k)$ of \problemcoverage where $G'$ is depicted in Figure \ref{figure:topologicgraphproblemcoverage}. $G'$ contains $G$ as a subgraph, plus fresh nodes $v_1, \dots, v_k$ and $s_1, \dots, s_k$. A UAV can move from any node of $G$ to $v_1$ and vice-versa.

		\begin{figure}[t]
			\begin{center}
				\newcommand{\ymax}{5.3}
				\newcommand{\ybase}{2}
				%\scalebox{0.8}
				{\begin{tikzpicture}[scale=0.9]
				\draw[copyofG] (-0.2,\ybase-0.3) rectangle (8.2,3.5);
				\draw (4,\ybase) node[draw] (nb) {\base};
				\draw (1,4) node (s1) {$s_1$};
				\draw (3,4) node (s2) {$s_2$};
				\draw (7,4) node (sk) {$s_k$};
				\draw (1,3.2) node (t1) {$c_1$};
				\draw (3,3.2) node (t2) {$c_2$};
				\draw (7,3.2) node (tk) {$c_k$};
				\draw (1,4.8) node (v1) {$v_1$};
				\draw (3,4.8) node (v2) {$v_2$};
				\draw (7,4.8) node (vk) {$v_k$};
				\draw (8,5.3) node (all) {all};
				\draw[->] (t1) -- (s1);
				\draw[->] (t2) -- (s2);
				\draw[->] (tk) -- (sk);
				\draw[,communication] (s1) -- (s2);
				\draw[,communication] (s2) -- (4.5,4);
				\draw[dotted,color=blue!50!white] (4.5,4) -- (5.5,4);
				\draw[,communication] (sk) -- (5.5,4);
				\draw[,communication] (t1) -- (t2);
				\draw[,communication] (t2) -- (4.5,3.2);
				\draw[dotted,color=blue!50!white] (4.5,3.2) -- (5.5,3.2);
				\draw[,communication] (tk) -- (5.5,3.2);
				\addCommunicationB{s1}
				\draw[->] (s1) -- (v1);
				\draw[->] (s2) -- (v2);
				\draw[->] (sk) -- (vk);
				\draw[<->] (v1) to[bend right] (0.1,3.4);
				\path (vk) edge [loop right] (vk);
				\draw[->] (nb) to [loop right] (nb);
				\draw [->, rounded corners = 8pt] (vk) -- (7,\ymax) -- (-1,\ymax) -- (-1, \ybase)  -- (nb);
				\draw (v1) -- (1,\ymax);
				\draw (v2) -- (3,\ymax);
				\draw[,communication] (vk) -- (all);
				\node at (4, 2.5) {copy of $G$};
				\end{tikzpicture}}
				\caption{Topologic graph $G'$ of the $\problemcoverage$-instance constructed from the $\problemreachability$-instance.}
				\label{figure:topologicgraphproblemcoverage}
			\end{center}
		\end{figure}
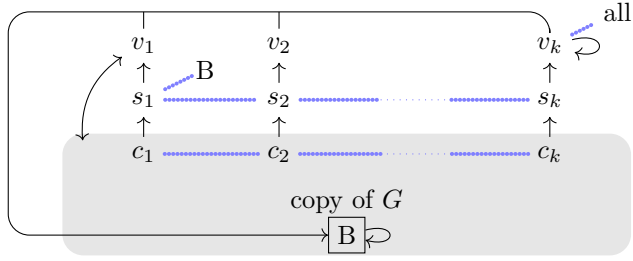

	 Node $s_1$ can communicate with the base $B$ and node $v_k$ can communicate with all nodes of $G'$. Now we prove that the $k$ UAVs can progress to the configuration ($c_1,\dots,c_k$) in $G$  if and only if there exists a covering execution in~$G'$.

($\Rightarrow$) If the UAVs are in the configuration ($c_1,\dots,c_k$) then they can progress in one step to configuration ($s_1,\dots,s_k$). Then, they have no choice but progress to the configuration ($v_1,\dots,v_k$). 
Once in this configuration, the UAV placed on the node $v_k$ can communicate with any UAV, placed on any node, and to the base $B$. Actually that UAV will stay at $v_k$.
Meanwhile the UAV placed on the node $v_1$ will visit all unvisited nodes of $G$  and come back to $v_1$ while keeping communication to the base through the UAV placed on $v_k$. Meanwhile, UAVs placed on $v_2,\dots,v_{k-1}$ come back to $B$. Finally, when all the nodes
have been visited, both UAVs on $v_1$ and $v_k$ come back to $B$.

\newcommand{\timesk}{t_{s_k}}

($\Leftarrow$) If there exists a covering execution of the whole graph $G'$, it means all nodes have been visited. In particular, node $s_k$ has been visited and let us consider the first time $\timesk$ when $s_k$ is visited. Time $\timesk-1$ denotes the time just before $\timesk$.

\begin{fact}
	\label{fact:nodesoutsideGprimeunvisited}
 At time $\timesk-1$, no node $v_i$ were visited and no node $s_i$ were visited.
\end{fact}
\begin{proof}
 Suppose by contradiction that a node $v_i$ was visited by some UAV before $\timesk$, then the only possibility such a UAV to communicate to the base is that there is also a UAV at $v_k$ at time $\timesk$. But then, it means that $s_k$ was visited strictly before $\timesk$, leading to a contradiction. Thus, no node $v_i$ were visited at time $\timesk$ (thus at time $\timesk-1$).
 
 As no node $v_i$ are visited before $\timesk$, no node $s_i$ are visited before $\timesk-1$.
\end{proof}

\begin{fact}
	\label{fact:reachc1ck}
	At time $\timesk-1$, the configuration is $(c_1, \dots, c_k)$.
\end{fact}

\begin{proof}
 At time $\timesk$, as the UAV at $s_k$ needs to communicate, the only possibility is that the configuration is $(s_1, \dots, s_k)$. Thus, the only possibility is that configuration is $(c_1, \dots, c_k)$.
\end{proof}

Facts~\ref{fact:nodesoutsideGprimeunvisited} implies that implies that the prefix from time 0 to time $\timesk-1$ of the covering execution is an execution in $G$. Fact~\ref{fact:reachc1ck} implies that subexecution reaches $(c_1,\dots,c_k)$.
\end{proof}

\subsection{NP lower bound for bounded problems}
In this subsection, we reduce the square tiling problem which is NP-complete (Theorem~\ref{theorem:squaretilingproblemPSPACEcomplete}) to $\problemboundedreachability$ and $\problemboundedcoverage$.

\begin{theorem}
	\problemboundedreachability is NP-hard.
	\label{theorem:boundedreachbilityisnphard}
\end{theorem}

\begin{figure}[t]
	\begin{center}
		\newcommand{\xdots}{4.3}
		\newcommand{\dotswidth}{0.2}
		\newcommand{\xcopyk}{5.6}
		\newcommand{\xrightnode}{7.6}
		\newcommand{\coordinategeek}[2]{\hspace{-4mm}\raisebox{3mm}{\tiny{$(#1,#2)$}}}
		\scalebox{0.75}
		{\begin{tikzpicture}
			\draw[fill=\colorellipse] (1,3) ellipse (0.9 and 1.6);
			\draw[fill=\colorellipse] (3,3) ellipse (0.9 and 1.6);
			% \draw (5,3) ellipse (0.9 and 2.1);
			\draw[fill=\colorellipse] (\xcopyk,3) ellipse (0.9 and 1.6);
			\draw[dashed,rounded corners=8pt] (-0.2,1.2) rectangle (\xcopyk+1.2,4.8);
			\draw (4,-0.3) node[draw] (nb) {B};
			\draw (1,0.75) node (n1) {$\bottomtile{1}$\coordinategeek{0}{1}};
			\draw (3,0.75) node (n2) {$\bottomtile{2}$\coordinategeek02};
			\draw (\xcopyk,0.75) node (nk) {$\bottomtile{k}$\coordinategeek0k};
			\draw (-1,0.75) node (lb) {$\lefttile{bot}$\coordinategeek00};
			\draw (-1,2.2) node (l1) {$\lefttile{1}$\coordinategeek 10};
			%  \draw (-1,2.2) node (l2) {$\lefttile{2}$};
			\draw (-1,3.8) node (lk) {$\lefttile{k}$\coordinategeek k0};
			\draw (-1,5.25) node (lt) {$\lefttile{top}$\coordinategeek{k+1}0};
			\draw (\xrightnode,0.75) node (rb) {$\righttile{bot}$\coordinategeek0{k+1}};
			\draw (\xrightnode,2.2) node (r1) {$\righttile{1}$\coordinategeek1{k+1}};
			%  \draw (9,2.2) node (r2) {$\righttile{2}$};
			\draw (\xrightnode,3.8) node (rk) {$\righttile{k}$\coordinategeek k{k+1}};
			\draw (\xrightnode,5.25) node (rt) {$\righttile{top}$\coordinategeek {k+1}{k+1}};
			\draw (1,5.25) node (t1) {$\toptile{1}$\coordinategeek{k+1}{1}};
			\draw (3,5.25) node (t2) {$\toptile{2}$\coordinategeek{k+1}{2}};
			\draw (\xcopyk,5.25) node (tk) {$\toptile{k}$\coordinategeek{k+1}{k}};
			\draw[->] (nb) -- (lb);
			\draw[->] (nb) -- (n1);
			\draw[->] (nb) -- (n2);
			\draw[->] (nb) -- (nk);
			\draw[->] (nb) -- (rb);
			\draw[->] (lb) -- (l1);
			% \draw[->] (l1) -- (l2);
			\draw[->,dashed] (l1) -- (lk);
			\draw[->] (lk) -- (lt);
			\draw[->] (rb) -- (r1);
			% \draw[->] (r1) -- (r2);
			\draw[->,dashed] (r1) -- (rk);
			\draw[->] (rk) -- (rt);
			\draw (1,3) node [text width = 1.8cm, text centered] {\scriptsize{copy$_{1}$ of T}};
			\draw (3,3) node [text width = 1.8cm, text centered] {\scriptsize{copy$_{2}$ of T}};
			\draw (4.25,3) node [text width = 1.8cm, text centered] {...};
			\draw (\xcopyk,3) node [text width = 1.8cm, text centered] {\scriptsize{copy$_{k}$ of T}};
			\draw[communication] (lb) -- (n1);
			\draw[communication] (n2) -- (n1);
			\draw[communication] (n2) -- (4,0.75);
			\draw[dotted,color=blue!50!white] (4,0.75) -- (4.5,0.75);
			\draw[communication] (nk) -- (4.5,0.75);
			\draw[,communication] (nk) -- (rb);
			\draw[communication] (l1) -- (0.8,2.2);
			\draw[dotted,color=blue!50!white] (1,2.2) -- (\xcopyk+0.2,2.2);
			\draw[communication] (r1) -- (\xcopyk+.2,2.2);
			%\draw[<-,communication] (l2) -- (0.2,2.2);
			%\draw[dotted,color=blue!50!white] (0.2,3.2) -- (7.8,3.2);
			% \draw[<-,communication] (r2) -- (7.8,2.2);
			\draw[communication] (lk) -- (0.8,3.8);
			\draw[dotted,color=blue!50!white] (0.8,3.8) -- (\xcopyk+.2,3.8);
			\draw[communication] (rk) -- (\xcopyk+.2,3.8);
			\draw[communication] (nk) -- (rb);
			\draw[->] (0.8,3.8) -- (t1);
			\draw[->] (1,3.8) -- (t1);
			\draw[->] (1.2,3.8) -- (t1);
			\draw[->] (3,3.8) -- (t2);
			\draw[->] (\xcopyk,3.8) -- (tk);
			\draw[->] (\xcopyk+.1,3.8) -- (tk);
			\draw[->] (n1) -- (0.9,2.2);
			\draw[->] (n1) -- (1.1,2.2);
			\draw[->] (n2) -- (2.8,2.2);
			\draw[->] (n2) -- (3,2.2);
			\draw[->] (n2) -- (3.2,2.2);
			\draw[->] (nk) -- (\xcopyk,2.2);
			\draw[communication] (lt) -- (t1);
			\draw[communication] (t2) -- (t1);
			\draw[communication] (t2) -- (4,5.25);
			\draw[dotted,color=blue!50!white] (4,5.25) -- (4.5,5.25);
			\draw[communication] (tk) -- (4.5,5.25);
			\draw[communication] (tk) -- (rt);
			\draw[->] (nb) to [out=330,in=300,looseness=8] (nb);
			\draw[communication] (nb) to[bend left] (lb);
			\draw[communication] (nb) to[bend left] (-2,1);
			\draw[communication] (-2,1) to[bend left] (l1);
			% \draw[->,communication] (-2,0) to[bend left] (l2);
			\draw[communication] (-2,1) to[bend left] (lk);
			\draw[communication] (-2,1) to[bend left] (lt);
			\end{tikzpicture}}
		\caption{Topologic graph of the bounded tiling problem reduction.}
		\label{bounded-reduction}
	\end{center}
\end{figure}
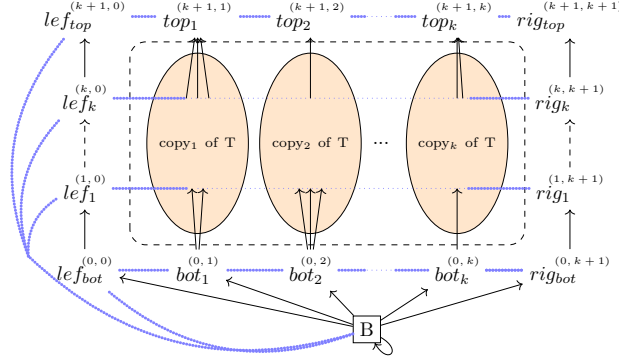

\begin{proof}

 The  proof is by polynomial reduction from the square tiling problem. From a instance $(T, \overrightarrow{\toptile{}},  \overrightarrow{\bottomtile{}},  \overrightarrow{\lefttile{}}, \overrightarrow{\righttile{}},k)$ of the square tiling problem, we will construct a $\problemboundedreachability$-instance $(G, c, k+2)$.
The topologic graph $G$ of Figure~\ref{bounded-reduction} looks like the one of Figure~\ref{reduction}. It uses the same conventions for 
 movements and communication. 
 %In this graph, however, the nodes $\bottomtile{1},\dots,\bottomtile{k},\toptile{1},\dots,\toptile{k}$ are not tiles anymore, but simply colors, respectively of the top and the bottom edges of the square. 
 In this graph, the ellipses are now surrounded by nodes  $\bottomtile{1},\dots,\bottomtile{k},\toptile{1},\dots,\toptile{k},\lefttile{1},\dots,\lefttile{k},\righttile{1},\dots,\righttile{k}$ that represent the bottom, top, left and right edge colors of the $k \times k$-square. More precisely:
 
 \begin{itemize}
 	\item  $\lefttile{i}$ (resp. $\righttile{i}$) is $\communication$-connected to all tiles of the first (resp. $k^{\text{th}}$) copy whose left (resp. right) color is $\lefttile{i}$ (resp. $\righttile{i}$);
 	\item  $\bottomtile{i}$ (resp. $\toptile{i}$) is $\rightarrow$-connected to (is $\rightarrow$-reachable from) all tiles of the $i^{\text{th}}$ copy whose bottom (resp. top) color is $\bottomtile{i}$ (resp. $\toptile{i}$).
 \end{itemize}

Actually the idea of the reduction is similar to the proof of Theorem~\ref{theorem:reachabilitypspacehard} except that now, two extra UAVs runs on the $\lefttile{i}$- and $\righttile{i}$-lanes to control both the left- and right- boundary color constraints and the vertical size of the square tiling.
 More formally $G$ = ($V$,$\rightarrow$,$\communication$) is defined by:
 \begin{itemize}
 	\item $V$ is the disjoint union of the sets $\set{B}$, $T \times \set{1, \dots, k}$ and $\set{(i, j) \in \set{0, \dots, k+1}^2 \suchthat \text{either ($i=0$ or $i=k+1$) or ($j=0$ or $j=k+1$)}}$;
 	\item $\rightarrow$ is the union of the sets  $\set{(B, (0, i)) \suchthat i \in \set{0, \dots, k+1}}$,  $\set{((t, i), (t', i)) \suchthat up(t) = down(t')}$, $\set{((0,i), (t, i)) \suchthat \bottomtile{i} = down(t)}$, $\set{((t, i), (k+1, i)) \suchthat \toptile{i} = up(t)}$  for all $i \in \set{1, \dots, k}$, $\{((i,0),(i+1,0)), i \in \{0,\dots,k\}\}$ and  $\{((i,k+1),(i+1,k+1), i \in \{0,\dots,k\} \}$;
 	\item $\communication$ is the union of the sets $\set{(B, (i, 0)) \suchthat i \in \set{0, \dots, k+1}}$, and $\set{((t, i), (t', i+1)) \suchthat right(t) = left(t')}$ for all $i \in \set{1, \dots, k-1}$, $\set{((i,j),(i,j+1)) \suchthat \text{($i=0$ or $i=k+1$) and $j\in\set{0, \dots, k}$}}$, $\set{((i, 0), (t, 1)) \suchthat \lefttile{i} = left(t)}$ and $\set{((i, k+1), (t, k)) \suchthat \righttile{i} = right(t)}$.
 	
 \end{itemize}

 The \problemboundedreachability instance is $(G, c, k+2)$ where $c$ is the configuration $((k+1,0), \dots (k+1, k+1))$ ($(\lefttile{top}, \toptile{1}, \dots, \toptile{k}, \righttile{top})$ in Figure~\ref{bounded-reduction}).
\end{proof}

\begin{theorem}
	\problemboundedcoverage is NP-hard.
\end{theorem}

\begin{proof}
	The idea is similar than for Theorem \ref{theorem:coveragepspacehard}. We proceed by polynomial time reduction from the square tiling problem. First we apply the reduction  given in the proof of Theorem~\ref{theorem:boundedreachbilityisnphard}: from an instance $(T, \overrightarrow{\toptile{}},  \overrightarrow{\bottomtile{}},  \overrightarrow{\lefttile{}}, \overrightarrow{\righttile{}},k)$  we obtain an \problemboundedreachability-instance  of the form $(G, c, k+2)$ (as depicted in Figure~\ref{bounded-reduction}), where $c = ((k+1, 0), \dots, (k+1, k+1))$.
	 Notice that these instances are such that all executions are of length at most $k+2$. Therefore, the existence of an execution is the same that the existence of an execution of length at most $k+2$.
	Thus, the same construction of Figure \ref{figure:topologicgraphproblemcoverage} (just $k+2$ instead of $k$) is sound. Indeed, from $(G, c, k+2)$, we construct the instance $(G', k+2, \ell)$ where $\ell$ is the sum of $k+2$ (the number of steps to reach the configuration $c$), 2 (the two steps to reach the configuration $(v_1, \dots, v_{k+2})$), $2 \times |G| + 1$ (an upper bound of the number of steps for the drone in $v_1$ to visit the unvisited node in the subgraph $G$ and to come back to the base).%,  = k+4 + 2\times |T| \times k+1$, where $k+4$ is the time needed to arrive at $(v_1,\dots,v_{k+2})$ and $2\times |T| \times k+1$ is an upper bound of the time needed for the drone in $v_1$ to cover the rest of the graph and to come back to $B$ afterwards.
\end{proof}

\subsection{Restrictions to neighbor-communicable graphs}

In this subsection, we prove that the lower bounds still hold for neighbor-communicable graphs.

\begin{theorem}
	\label{theorem:reachabilitycoveragePSPACEhardneighbour-communicable}
	\problemreachability and \problemcoverage are PSPACE-hard even when restricted to neighbor-communicable topologic graphs.
\end{theorem}

\begin{proof}
	For \problemreachability, the proof is similar to the proof of Theorem \ref{theorem:reachabilitypspacehard}: we just slightly modify graph $G$ of Figure \ref{reduction} as follows.  In order to prevent a UAV at $bot_2,\dots,bot_k$ to communicate directly with $B$, we add an intermediate node $m_i$ between $B$ and each $bot_i$ ($B \rightarrow bot_i$ becomes $B \rightarrow m_i \rightarrow bot_i$) for $i = 1, \dots, k$. We also add a communication edge $v \communication v'$ whenever $v \rightarrow v'$. In rest of the proof, the new graph is still noted $G$. 
	
	For \problemcoverage, the construction given in Figure~\ref{figure:topologicgraphproblemcoverage} with the new graph $G$ does not work. Indeed, all nodes may be visited  although $c_1, \dots, c_k$ was not reached: maybe $v_1$ and $v_k$ are reached by two lines of UAVs connected to the base, making the coverage of the full graph possible.

		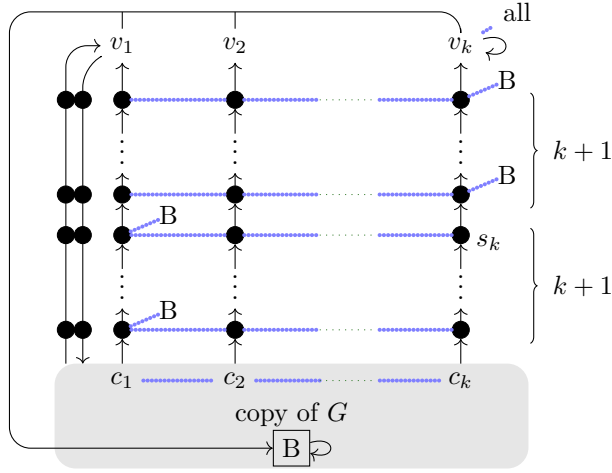
\begin{figure}[t]
		\begin{center}
			\newcommand{\xretourgauche}{-1}
			\tikzstyle{circlenode} =  [circle, draw=black, fill=black, inner sep=0.8mm]
			\newcommand{\ligne}[1]{
				\node[circlenode]  (s1) at (0,#1)  {};
				\node[circlenode]  (s1) at (0.3,#1)  {};
				\node[circlenode]  (s1) at (1,#1)  {};
				\node[circlenode] (s2) at (3,#1)   {};
				\node[circlenode]  (sk) at (7,#1)  {};
				\draw[->] (1, #1-0.5) -- (s1);
				\draw[->] (3, #1-0.5) -- (s2);
				\draw[->] (7, #1-0.5) -- (sk);
				\draw[communication] (s1) -- (s2);
				\draw[communication] (s2) -- (4.5,#1);
				\draw[dotted,color=OliveGreen] (4.5,#1) -- (5.5,#1);
				\draw[communication] (sk) -- (5.5,#1);}
			
			\newcommand{\lignevdots}[1]{
				\node[]  (s1) at (1,#1)  {$\vdots$};
				\node[] (s2) at (3,#1)   {$\vdots$};
				\node[]  (sk) at (7,#1)  {$\vdots$};
				\draw[->] (1, #1-0.7) -- (s1);
				\draw[->] (3, #1-0.7) -- (s2);
				\draw[->] (7, #1-0.7) -- (sk);
			}
			
			\newcommand{\vnodeY}{8.7}
			\newcommand{\posBaseY}{2.75}
			\scalebox{1}
			{	\begin{tikzpicture}[xscale=0.75,yscale=0.9]
				\draw[copyofG] (-0.2,\posBaseY-0.3) rectangle (8.2,4);
				\draw (4,\posBaseY	) node[draw] (nb) {B};
				\draw (1,3.75) node (t1) {$c_1$};
				\draw (3,3.75) node (t2) {$c_2$};
				\draw (7,3.75) node (tk) {$c_k$};
				
				\draw[communication] (t1) -- (t2);
				\draw[communication] (t2) -- (4.5,3.75);
				\draw[dotted,color=OliveGreen] (4.5,3.75) -- (5.5,3.75);
				\draw[communication] (tk) -- (5.5,3.75);

				\ligne {4.5}
				\addCommunicationB {s1}
				\lignevdots {5.25}
				\ligne {5.9}
				\addCommunicationB {s1}
				
				\node at (7.5, 5.8) {$s_k$};
				
				\draw[decoration={brace,mirror,raise=5pt},decorate]
				(8,4.3) -- node[right=10pt] {$k+1$} (8, 6);

				\ligne {6.5}
				\addCommunicationB {sk}
				%			\ligne {10.5}
				%			\addCommunicationB {sk}
				\lignevdots {7.25}
				\ligne {7.9}
				\addCommunicationB {sk}

				\draw[decoration={brace,mirror,raise=5pt},decorate]
				(8,6.3) -- node[right=10pt] {$k+1$} (8,8);

				\draw (1,\vnodeY) node (v1) {$v_1$};
				\draw (3,\vnodeY) node (v2) {$v_2$};
				\draw (7,\vnodeY) node (vk) {$v_k$};
				\draw (8,\vnodeY+0.5) node (all) {all};
				\draw[->] (s1) -- (v1);
				\draw[->] (s2) -- (v2);
				\draw[->] (sk) -- (vk);
				\draw[->, rounded corners = 8pt] (v1) -- (0.3, \vnodeY-0.3) -- (0.3,4);
				\draw[<-, rounded corners = 8pt] (v1) -- (0, \vnodeY) -- (0,4);
				
				\path (vk) edge [loop right] (vk);
				\draw[->] (nb) to [loop right] (nb);
				\draw [->, rounded corners = 8pt] (vk) -- (7,\vnodeY+0.5) -- (\xretourgauche,\vnodeY+0.5) -- (\xretourgauche, \posBaseY)  -- (nb);
				\draw (v1) -- (1,\vnodeY+0.5);
				\draw (v2) -- (3,\vnodeY+0.5);
				
				\draw[communication] (vk) -- (all);
				\node at (4, 3.25) {copy of $G$};
				\end{tikzpicture}}
			\caption{Topologic graph of the $\problemcoverage$-instance constructed from the $\problemreachability$-instance for the case of neighbor-communicable topologic graph.}
			\label{figure:topologicgraphproblemcoveragefixedforneighborcommunicable}
		\end{center}
	\end{figure}

	The corrected construction is given in Figure \ref{figure:topologicgraphproblemcoveragefixedforneighborcommunicable}. When configuration $(c_1,...,c_k)$ is reached, the UAVs go through a first layer of length $k+1$ in which the first UAV can communicate with $B$. Then they go through another layer of length $k+1$ in which the $k^{th}$ UAV can communicate with $B$. This way, it is mandatory that all UAVs move at the same time to visit $(v_1,\dots,v_k)$. Once the $k^{th}$ UAV is at $v_k$, all UAVs can communicate with $B$ wherever they are, so they can visit remaining states in the copy of $G$. Now let us prove that $(c_1, \dots, c_k)$ is reachable in $G$ iff it is possible to cover all nodes in $G'$.

	($\Rightarrow$) If $(c_1, \dots, c_k)$ is reachable in $G$, then we extend the execution to reach $(v_1, \dots, v_k)$ and by the same trick as in Figure \ref{figure:topologicgraphproblemcoverage}, the UAV that reaches $v_1$ visits all the remaining unvisited nodes in $G$. Thus, we extend the execution for covering all nodes in $G'$.
	
	$(\Leftarrow)$ Suppose all nodes are visited in $G'$.
	In particular, $v_1$ and $v_k$ are visited. Let us consider the first moment $t_{v_i}$ when a node $v_i$ is visited.
	
	\begin{fact}
		\label{fact:nodesv1vk}
		At that first moment, the configuration is $(v_1, \dots, v_k)$.
	\end{fact}
	
	\begin{proof}
		Let us prove that there is a UAV at $v_k$. Suppose that at that moment there is no UAV at $v_k$. Due to the topological graph $G'$, the UAV at $v_i$ is disconnected from the base since  nodes that communicate directly to $B$ are too far from $v_i$: indeed, the top $k+1$-grid is too long and, for $i=1$, the path on left between $v_1$ and the copy of $G$ is too long. Contradiction.
		
		The UAV at $v_k$ came from the unique $2k+2$-long path from $c_k$ to $v_k$.
		Actually, $k+1$ steps before - let us call this moment $t_{s_k}$, she was on $s_k$. But at that time, due to the topological graph, there are $k$ UAVs on the row containing $s_k$, otherwise the UAV at $s_k$ would have been disconnected from the base (the bottom $k+1$-grid is too long).
		
		So $k+1$ times later $t_{s_k}$, all the $k$ UAVs are at $(v_1, \dots, v_k)$.
	\end{proof}

	Taking Fact \ref{fact:nodesv1vk} as granted, we consider time $t$ that is $2k+2$ steps before and we clearly have the following fact.
	
	\begin{fact}
		At time $t$, the configuration is $(c_1, \dots, c_k)$.
	\end{fact}

	Moreover, the following fact holds.
	
	\begin{fact}
		At time $t$, no node outside $G$ were visited.
	\end{fact}
	
	\begin{proof}
		By contradiction, if some node outside $G$ were visited, it means that some UAV went out the copy of $G$. By definition of $G'$, it would mean that a node $v_i$ would have been visited, before time $t$, hence strictly before $t_{v_i}$. Contradiction.
	\end{proof}

To sum up, the prefix of the execution from $(B, \dots, B)$ to $(c_1, \dots, c_k)$ is fully inside the copy of $G$. So $(c_1, \dots, c_k)$ is reachable in $G$.
\end{proof}

\begin{theorem}
Both \problemboundedreachability and \problemboundedcoverage are NP-hard when restricted to neighbor-communicable topological graphs.
\end{theorem}

\begin{proof}
	For \problemboundedreachability, we use the same construction depicted in Figure \ref{bounded-reduction} except that we add intermediate nodes (as in the proof of Theorem \ref{theorem:reachabilitycoveragePSPACEhardneighbour-communicable}) between $B$ and $lef_{bot}, bot_1, \dots, bot_k, rig_{bot}$ and edges $v \communication v'$ are added whenever $v \rightarrow v'$. By still calling $G$ the obtained graph, the $\problemboundedreachability$-instance is $(G, c, k+3)$, where $c = ((k+1, 0), \dots, (k+1, k+1))$. The bound is now $k+3$ instead of $k+2$ because of the intermediate nodes.

	 For \problemboundedcoverage, we use the same idea that in Theorem~\ref{theorem:boundedreachbilityisnphard} but the construction given in Figure \ref{figure:topologicgraphproblemcoveragefixedforneighborcommunicable}.  The bound $\ell$ is the sum of $k+3$ (the number of steps for reaching $c$ in $G$), $2\times(k+2)$ (the number of steps to reach the configuration $(v_1, \dots, v_k)$) and $2 \times (k+2) \times |G| + 1$ (the number of steps for the first UAV finishing the visit of all remaining unvisited nodes; $2 \times (2k+2)$ corresponds to the number of steps in the two left-most paths in Figure~\ref{figure:topologicgraphproblemcoveragefixedforneighborcommunicable} for the back and forth between $v_1$ and nodes of $G$).
\end{proof}

\section{Related work}
\label{section:relatedwork}
As shown in the survey by Chen et al. \cite{ChenSurvey}, many coverage problems have been addressed by using analytic techniques. For instance, in \cite{DBLP:conf/icc/Yanmaz12} and \cite{teacy2010maintaining}, they also address UAVs that should cover an area while staying connected to the base, but solve this problem with specific path planning algorithms. The algorithms they provide are not proven formally but tested experimentally.

That is why we advocate for formal methods, that have already been applied to generate plans for robots and UAVs. For instance, model checking has been applied to robot planning (see \cite{DBLP:conf/iros/LacerdaPH14}) and to UAVs. Humphrey \cite{Humphrey2013} shows how to use LTL (linear-temporal logic) model checking for capturing response and fairness properties in cooperation (for instance, if a task is requested then it is eventually performed). Model checking has also been used to verify pre-programmed UAVs \cite{webster2011formal}.% In \cite{DBLP:journals/corr/abs-1003-0381}, they discuss CTL model checking for checking properties. CTL is not suitable for our purpose because we need to express the existence of \emph{one} path along which UAVs stay connected and eventually have covered all the locations and have come back to the base location.

Bodin et al. \cite{IJCAI2018demodrones} treat a similar problem except that the UAVs cover the graph without returning to the base. If we remove the return to the base constraint, we claim that all our complexity results still hold. They provide an implementation by describing the problem in PDDL (Planning Domain Description Language) and then run the planner FS  (Functional Strips) \cite{DBLP:conf/ijcai/FrancesRLG17}. Both \problemreachability and \problemcoverage may be expressed in MA-STRIPS \cite{DBLP:conf/aips/BrafmanD08}, that is a multi-agent variant of STRIPS (Stanford Research Institute Problem Solver) in which actions for each agent can be described independently. The representation in multi-agent planning languages is especially efficient when actions of the different agents are independent and when they required to coordinate not so often. However, as the agents should maintain connection, it requires a lot of coordination.

Interestingly Murano et al. \cite{DBLP:conf/prima/MuranoPR15} advocate for a graph-theoretic representations of states, that is, by giving locations to agents as we do in Definition~\ref{definition:configuration}. Aminof et al (\cite{DBLP:conf/atal/AminofMRZ16,DBLP:conf/atal/Rubin15}) propose a very general formalism to specify LTL and MSO (monadic second-order logic) properties which is expressive enough to express connectivity between agents with an MSO formula. Indeed, linear temporal operators enable to express that any vertex should be visited in the future and the connectivity invariant. MSO on the topological graph enables to express the connectivity as a fix point (the subgraph made up of the UAVs and the base is connected).
%\newcommand{\droneat}[2]{at(#1,#2)}
%
%$$\left(\bigwedge_{v \in V} F \bigvee_{i \in \set{1, \dots, \nbdrones}} \droneat i v \right) \land  $$
%
They provide an algorithm for parametrized verification in the sense that they check a temporal property in a class of graphs. This is relevant for partially-known environments. The algorithm described is non elementary and therefore not usable in practice. We nevertheless claim that studying fragments of it is relevant, and our paper seems to be a relevant fragment.

%\todo{le problème de tâche à étudier (article qu'a lu Eva)}
%
%\todo{pt un lien avec le jeu de taquin}
%
%
%\todo{parler de energy games de Nicolas Markey et al. pour la batterie}
%
%\todo{citer Holger Hermanns, Timed Automata and Battery Kinetics in Thermosphere (conseillé par Sophie car ils gèrent la batterie)}

%Holger Hermanns, Saarland University
%
%Timed Automata and Battery Kinetics in Thermosphere
%
%For a satellite orbiting the earth all resources are sparse and the most critical resource of all is power. It is therefore crucial to have detailed knowledge on how much power is available for an energy harvesting satellite in orbit at every time – especially when in eclipse, where it draws its power from onboard batteries. This presentation addresses this problem by a two-step procedure to perform task scheduling for satellites orbiting the earth. It combines cost-optimal reachability analyses of priced timed automata networks with a realistic kinetic battery model capable of capturing capacity limits as well as stochastic fluctuations. The procedure has been put in place for the automatic and resource-optimal day-ahead scheduling of GomX–3, a power-hungry nanosatellite orbiting the earth in thermosphere. We explain how this approach has overcome existing problems, has led to improved designs, and has provided new insights.
%
%This presentation covers joint work with Morten Bisgaard, David Gerhardt (both from Gomspace AS, Denmark), Jan Krčál, Gilles Nies, and Marvin Stenger (all from Saarland University, Germany)
%
%https://www.cse.iitb.ac.in/~krishnas/averts2016/index.html#speakers

\section{Conclusion}
\label{section:conclusion}
On the theoretical side,  we introduced the multi-agent planning problems decision problems - namely \problemboundedcoverage, \problemboundedreachability, \problemcoverage, \problemreachability - that could become standard problems for proving that other multi-agent decision problems are NP-hard or PSPACE-hard. In some sense, this paper could be the starting point of a theory of multi-agent problems in complexity theory as constraint logic \cite{DBLP:conf/coco/DemaineH08} is for games.

Up to now, it is unknown whether our decision problems remain hard when the $\rightarrow$-relations become symmetric. We think this open issue is important since symmetric {$\rightarrow$-relations} (if UAVs can go from $v$ to $v'$, they can also come back from $v'$ to $v$) are relevant for practical applications. We also plan to study the \emph{parametrized complexity} \cite{DBLP:series/mcs/DowneyF99} of our problems - parameters could be the treewidth of the topological graph, the number of UAVs.

Interestingly, we plan to generalize to decentralized versions of our problems and to dynamic environments.  Instead of generating sequences of actions, we will have to generate strategies as in ATL (alternating-time temporal logic) \cite{DBLP:journals/corr/abs-1006-1414}. As UAVs stay connected, we may suppose that when information is gained, it is common knowledge and that all actions, especially sensing actions, are public \cite{DBLP:conf/atal/BelardinelliLMR17}. We also aim at using a high-level dedicated formal logic to express objectives, such as the language proposed in \cite{DBLP:conf/atal/Rubin15} and \cite{DBLP:conf/atal/AminofMRZ16}.

\paragraph{Acknowledgment. } This work was essentially done during the summer of~2017 with the internship of Eva Soulier. We thank François Bodin for all the discussions for the development of the definition of the topological graph presented in this document.  We thank Sophie Pinchinat and Ocan Sankur for their comments.

%\todo{citer multi-agent planning et apprendre}

%\todo{parler de complexité paramétrée (nombre de drones)}

%A worth advantage of model checking over planning is that it captures infinite behaviors. For reachability goals, model checking and planning are c \todo{Dire que l'on prévoit le merge de la comunauté planning et de la communauté model checking}
%Note that LTL is expressive enough to capture other scenarios in which particular UAVs can reached some regions or when some region $\node$ should be visited before another region $\node'$.

%, an do not support infinite executions. \todo{scaling}

%\todo{optimisation of the model checking: connectivity (formule temporelle pour exprimer la connectivité ? un algo pour tester la connectivité appelé par le model checker)}

%\newpage

\bibliographystyle{alpha}
\bibliography{biblio}

\newcommand{\etalchar}[1]{$^{#1}$}
\begin{thebibliography}{TNMP10}

\bibitem[AMRZ16]{DBLP:conf/atal/AminofMRZ16}
Benjamin Aminof, Aniello Murano, Sasha Rubin, and Florian Zuleger.
\newblock Automatic verification of multi-agent systems in parameterised
  grid-environments.
\newblock In {\em Proceedings of the 2016 International Conference on
  Autonomous Agents {\&} Multiagent Systems, Singapore, May 9-13, 2016}, pages
  1190--1199, 2016.

\bibitem[BCQS18]{IJCAI2018demodrones}
François Bodin, Tristan Charrier, Arthur Queffelec, and Fran\c{c}ois
  Schwarzentruber.
\newblock Generating plans for cooperative connected uavs (demo).
\newblock In {\em Proceedings of the 27th International Joint Conference on
  Artificial Intelligence (IJCAI) and the 23rd European Conference on
  Artificial Intelligence (ECAI), Stockholm, 13-19 July 2018}, 2018.

\bibitem[BD08]{DBLP:conf/aips/BrafmanD08}
Ronen~I. Brafman and Carmel Domshlak.
\newblock From one to many: Planning for loosely coupled multi-agent systems.
\newblock In {\em Proceedings of the Eighteenth International Conference on
  Automated Planning and Scheduling, {ICAPS} 2008, Sydney, Australia, September
  14-18, 2008}, pages 28--35, 2008.

\bibitem[BLMR17]{DBLP:conf/atal/BelardinelliLMR17}
Francesco Belardinelli, Alessio Lomuscio, Aniello Murano, and Sasha Rubin.
\newblock Verification of multi-agent systems with imperfect information and
  public actions.
\newblock In {\em Proceedings of the 16th Conference on Autonomous Agents and
  MultiAgent Systems, {AAMAS} 2017, S{\~{a}}o Paulo, Brazil, May 8-12, 2017},
  pages 1268--1276, 2017.

\bibitem[Boa97]{Boas97theconvenience}
Peter Van~Emde Boas.
\newblock The convenience of tilings.
\newblock In {\em In Complexity, Logic, and Recursion Theory}, pages 331--363.
  Marcel Dekker Inc, 1997.

\bibitem[Byl94]{DBLP:journals/ai/Bylander94}
Tom Bylander.
\newblock The computational complexity of propositional {STRIPS} planning.
\newblock {\em Artif. Intell.}, 69(1-2):165--204, 1994.

\bibitem[CZX14]{ChenSurvey}
Y.~Chen, H.~Zhang, and M.~Xu.
\newblock The coverage problem in uav network: A survey.
\newblock In {\em Fifth International Conference on Computing, Communications
  and Networking Technologies (ICCCNT)}, pages 1--5, July 2014.

\bibitem[DEG10]{DBLP:journals/corr/abs-1006-1414}
Catalin Dima, Constantin Enea, and Dimitar~P. Guelev.
\newblock Model-checking an alternating-time temporal logic with knowledge,
  imperfect information, perfect recall and communicating coalitions.
\newblock In {\em Proceedings First Symposium on Games, Automata, Logic, and
  Formal Verification, {GANDALF} 2010, Minori (Amalfi Coast), Italy, 17-18th
  June 2010.}, pages 103--117, 2010.

\bibitem[DF99]{DBLP:series/mcs/DowneyF99}
Rodney~G. Downey and Michael~R. Fellows.
\newblock {\em Parameterized Complexity}.
\newblock Monographs in Computer Science. Springer, 1999.

\bibitem[DH08]{DBLP:conf/coco/DemaineH08}
Erik~D. Demaine and Robert~A. Hearn.
\newblock Constraint logic: {A} uniform framework for modeling computation as
  games.
\newblock In {\em Proceedings of the 23rd Annual {IEEE} Conference on
  Computational Complexity, {CCC} 2008, 23-26 June 2008, College Park,
  Maryland, {USA}}, pages 149--162, 2008.

\bibitem[FKP05]{DBLP:conf/icra/FainekosKP05}
Georgios~E. Fainekos, Hadas Kress{-}Gazit, and George~J. Pappas.
\newblock Temporal logic motion planning for mobile robots.
\newblock In {\em Proceedings of the 2005 {IEEE} International Conference on
  Robotics and Automation, {ICRA} 2005, April 18-22, 2005, Barcelona, Spain},
  pages 2020--2025, 2005.

\bibitem[FRLG17]{DBLP:conf/ijcai/FrancesRLG17}
Guillem Franc{\`{e}}s, Miquel Ram{\'{\i}}rez, Nir Lipovetzky, and Hector
  Geffner.
\newblock Purely declarative action descriptions are overrated: Classical
  planning with simulators.
\newblock In {\em Proceedings of the Twenty-Sixth International Joint
  Conference on Artificial Intelligence, {IJCAI} 2017, Melbourne, Australia,
  August 19-25, 2017}, pages 4294--4301, 2017.

\bibitem[HTK00]{Harel:2000:DL:557365}
David Harel, Jerzy Tiuryn, and Dexter Kozen.
\newblock {\em Dynamic Logic}.
\newblock MIT Press, Cambridge, MA, USA, 2000.

\bibitem[Hum13]{Humphrey2013}
Laura~R. Humphrey.
\newblock {\em Model Checking for Verification in UAV Cooperative Control
  Applications}, pages 69--117.
\newblock Springer Berlin Heidelberg, Berlin, Heidelberg, 2013.

\bibitem[KGLR18]{doi:10.1146/annurev-control-060117-104838}
Hadas Kress-Gazit, Morteza Lahijanian, and Vasumathi Raman.
\newblock Synthesis for robots: Guarantees and feedback for robot behavior.
\newblock {\em Annual Review of Control, Robotics, and Autonomous Systems},
  1(1):null, 2018.

\bibitem[Lev73]{levin1973}
L.~A. Levin.
\newblock Universal sequential search problems.
\newblock {\em Problems of Information Transmission}, 9(3):265--266, 1973.

\bibitem[LP98]{DBLP:books/daglib/0096996}
Harry~R. Lewis and Christos~H. Papadimitriou.
\newblock {\em Elements of the theory of computation, 2nd Edition}.
\newblock Prentice Hall, 1998.

\bibitem[LPH14]{DBLP:conf/iros/LacerdaPH14}
Bruno Lacerda, David Parker, and Nick Hawes.
\newblock Optimal and dynamic planning for markov decision processes with
  co-safe {LTL} specifications.
\newblock In {\em 2014 {IEEE/RSJ} International Conference on Intelligent
  Robots and Systems, Chicago, IL, USA, September 14-18, 2014}, pages
  1511--1516, 2014.

\bibitem[MPR15]{DBLP:conf/prima/MuranoPR15}
Aniello Murano, Giuseppe Perelli, and Sasha Rubin.
\newblock Multi-agent path planning in known dynamic environments.
\newblock In {\em {PRIMA} 2015: Principles and Practice of Multi-Agent Systems
  - 18th International Conference, Bertinoro, Italy, October 26-30, 2015,
  Proceedings}, pages 218--231, 2015.

\bibitem[MTS{\etalchar{+}}16]{DBLP:conf/aaai/MaTSKK16}
Hang Ma, Craig~A. Tovey, Guni Sharon, T.~K.~Satish Kumar, and Sven Koenig.
\newblock Multi-agent path finding with payload transfers and the
  package-exchange robot-routing problem.
\newblock In {\em Proceedings of the Thirtieth {AAAI} Conference on Artificial
  Intelligence, February 12-17, 2016, Phoenix, Arizona, {USA.}}, pages
  3166--3173, 2016.

\bibitem[Rub15]{DBLP:conf/atal/Rubin15}
Sasha Rubin.
\newblock Parameterised verification of autonomous mobile-agents in static but
  unknown environments.
\newblock In {\em Proceedings of the 2015 International Conference on
  Autonomous Agents and Multiagent Systems, {AAMAS} 2015, Istanbul, Turkey, May
  4-8, 2015}, pages 199--208, 2015.

\bibitem[Sav70]{DBLP:journals/jcss/Savitch70}
Walter~J. Savitch.
\newblock Relationships between nondeterministic and deterministic tape
  complexities.
\newblock {\em J. Comput. Syst. Sci.}, 4(2):177--192, 1970.

\bibitem[SvEB]{van1984boundedtiling}
M.W.P. Savelsbergh and Peter van Emde~Boas.
\newblock Bounded tiling, an alternative to satisfiability.
\newblock In {\em G. Wechsung (ed.), proc. 2nd Frege Memorial Conference,
  Schwerin, Sep 1984, Akademie Verlag, Mathematische Forschung vol. 20, 1984,
  pp. 401-407.}

\bibitem[TBR{\etalchar{+}}18]{DBLP:conf/aaai/TateoBRAB18}
Davide Tateo, Jacopo Banfi, Alessandro Riva, Francesco Amigoni, and Andrea
  Bonarini.
\newblock Multiagent connected path planning: Pspace-completeness and how to
  deal with it.
\newblock In {\em Proceedings of the Thirty-Second {AAAI} Conference on
  Artificial Intelligence, New Orleans, Louisiana, USA, February 2-7, 2018},
  2018.

\bibitem[TNMP10]{teacy2010maintaining}
WT~Luke Teacy, Jing Nie, Sally McClean, and Gerard Parr.
\newblock Maintaining connectivity in uav swarm sensing.
\newblock In {\em GLOBECOM Workshops (GC Wkshps), 2010 IEEE}, pages 1771--1776.
  IEEE, 2010.

\bibitem[Tur02]{DBLP:conf/jelia/Turner02}
Hudson Turner.
\newblock Polynomial-length planning spans the polynomial hierarchy.
\newblock In {\em Logics in Artificial Intelligence, European Conference,
  {JELIA} 2002, Cosenza, Italy, September, 23-26, Proceedings}, pages 111--124,
  2002.

\bibitem[Wan61]{Wang1961}
Hao Wang.
\newblock {\em Proving Theorems by Pattern Recognition, II}, pages 1--41.
\newblock AT\&T, 1961.

\bibitem[Wan90]{Wang1990}
Hao Wang.
\newblock {\em Proving Theorems by Pattern Recognition, II}, pages 159--192.
\newblock Springer Netherlands, Dordrecht, 1990.

\bibitem[WDM07]{DBLP:conf/aaai/WurmanDM07}
Peter~R. Wurman, Raffaello D'Andrea, and Mick Mountz.
\newblock Coordinating hundreds of cooperative, autonomous vehicles in
  warehouses.
\newblock In {\em Proceedings of the Twenty-Second {AAAI} Conference on
  Artificial Intelligence, July 22-26, 2007, Vancouver, British Columbia,
  Canada}, pages 1752--1760, 2007.

\bibitem[WFCJ11]{webster2011formal}
Matt Webster, Michael Fisher, Neil Cameron, and Mike Jump.
\newblock Formal methods for the certification of autonomous unmanned aircraft
  systems.
\newblock {\em Computer Safety, Reliability, and Security}, pages 228--242,
  2011.

\bibitem[Yan12]{DBLP:conf/icc/Yanmaz12}
Evsen Yanmaz.
\newblock Connectivity versus area coverage in unmanned aerial vehicle
  networks.
\newblock In {\em Proceedings of {IEEE} International Conference on
  Communications, {ICC} 2012, Ottawa, ON, Canada, June 10-15, 2012}, pages
  719--723, 2012.

\bibitem[YL13]{DBLP:conf/aaai/YuL13}
Jingjin Yu and Steven~M. LaValle.
\newblock Structure and intractability of optimal multi-robot path planning on
  graphs.
\newblock In {\em Proceedings of the Twenty-Seventh {AAAI} Conference on
  Artificial Intelligence, July 14-18, 2013, Bellevue, Washington, {USA.}},
  2013.

\end{thebibliography}

\end{document}